\documentclass[runningheads]{llncs}
\usepackage[utf8]{inputenc}
\usepackage{graphicx}
\usepackage{amsmath}
\usepackage{amssymb}
\usepackage{tikz}
\usepackage[ruled,vlined,onelanguage]{algorithm2e}
\usepackage{enumerate} 
\usepackage{hyperref}

\usepackage{wrapfig}

\usepackage{float}
\usepackage{pdflscape}
\usepackage{caption}
\usepackage{subcaption}
\usepackage{afterpage}

\newcommand{\node}[1][white]{
\begin{tikzpicture}
\tikzstyle{noeud}=[draw,circle,fill=#1,scale=0.7]
\node[noeud] (0) at ({0},{0}) {};
\end{tikzpicture}
}

\usetikzlibrary{cd}
\usetikzlibrary{fit}
\usetikzlibrary{backgrounds}
\usetikzlibrary{decorations.markings}
\usetikzlibrary{positioning}
\tikzset{->-/.style={decoration={
  markings,
  mark=at position .5 with {\arrow{>}}},postaction={decorate}}}
  \usetikzlibrary{shapes.geometric}
\tikzset{
    rubberduck/.style={
        draw=black,
        shape=isosceles triangle,
        fill=white,
        minimum height=1.5cm,
        minimum width=0.5cm,
        shape border rotate=#1,
        isosceles triangle stretches,
        inner sep=0pt,
    },
    rubber/.style={rubberduck=+90}}
    
 \DeclareMathOperator{\attr}{\mathcal{A}}

\DeclareMathOperator{\id}{\text{Id}}
\DeclareMathOperator{\isom}{\text{Isom}}
\DeclareMathOperator{\enc}{\text{Cipher}}
\DeclareMathOperator{\parent}{\mathcal{P}}
\DeclareMathOperator{\children}{\mathcal{C}}
\DeclareMathOperator{\depth}{\mathcal{D}}
\DeclareMathOperator{\leaves}{\mathcal{L}}

\DeclareMathOperator{\bags}{\mathbb{B}}
\DeclareMathOperator{\collections}{\mathbb{C}}

\newcommand\tf[1]{\xrightarrow{#1}}

\colorlet{lblue}{blue!50!white}
\colorlet{lred}{red!50!white}
\colorlet{lgreen}{green!50!white}
\colorlet{lpurple}{purple!50!white}
\colorlet{lorange}{orange!50!white}
\colorlet{lpink}{pink!50!white}
\colorlet{lbrown}{brown!50!white}
\colorlet{lyellow}{yellow!50!white}
\colorlet{lolive}{olive!50!white}

\spnewtheorem{deducrule}{Deduction Rule}{\bfseries}{\itshape}
\spnewtheorem{deducphase}{Deduction Phase}{\bfseries}{\itshape}

\begin{document}

\title{Isomorphic unordered labeled trees \\ up to substitution ciphering\thanks{Supported by European Union H2020 project ROMI.}}
%
%

\author{Florian Ingels\ \and
Romain Aza\"is}
\authorrunning{F. Ingels et R. Aza\"is}
%
\institute{Laboratoire Reproduction et D\'eveloppement des Plantes, Univ Lyon, ENS de Lyon, UCB Lyon 1, CNRS, INRAE, Inria, F-69342, Lyon, France\\
\email{\{florian.ingels,romain.azais\}@inria.fr}}
\maketitle              
\begin{abstract}
Given two messages -- as linear sequences of letters, it is immediate to determine whether one can be transformed into the other by simple substitution cipher of the letters. On the other hand, if the letters are carried as labels on nodes of topologically isomorphic unordered trees, determining if a substitution exists is referred to as marked tree isomorphism problem in the literature and has been show to be as hard as graph isomorphism. While the left-to-right direction provides the cipher of letters in the case of linear messages, if the messages are carried by unordered trees, the cipher is given by a tree isomorphism. The number of isomorphisms between two trees is roughly exponential in the size of the trees, which makes the problem of finding a cipher difficult by exhaustive search. This paper presents a method that aims to break the combinatorics of the isomorphisms search space. We show that in a linear time (in the size of the trees), we reduce the cardinality of this space by an exponential factor on average.

\\

\emph{This paper is eligible for best student paper award.}

\end{abstract}

\section{Introduction}
A \emph{simple substitution cipher} is a method of encryption that transforms a sequence of letters, replacing each letter from the original message by another letter, not necessarily taken from the same alphabet \cite{gardner1984codes}. 

Assume you have at your disposal two messages of the same length, and you want to determine if there exists a substitution cipher that transforms one message onto the other. This question is easily solved, as the cipher is induced by the order of letters. One letter after the other, you can build the cipher by mapping them, until (i) either you arrive at the end of the message, and the answer is Yes, (ii) either you detect an inconsistency in the mapping and the answer is No. Actually, this procedure induces an equivalence relation on messages of the same length: two messages are equivalent (isomorphic) if and only if there is a cipher that transforms one message onto the other. See Fig.~\ref{fig:sequence_cipher} for an illustration.

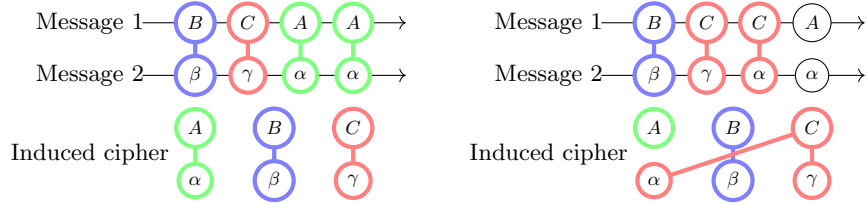
\begin{figure}[h!]
\centering

\begin{minipage}{0.5\textwidth}
\begin{tikzpicture}[xscale=0.7,yscale=0.7]
\tikzstyle{noeud}=[draw,circle,fill=white,scale=0.8]

\node at (-2,0) {Message 1};
\node at (-2,-1) {Message 2};
\node at (-2,-2.5) {Induced cipher};

\draw[->] (-1,0)--(4,0);
\draw[->] (-1,-1)--(4,-1);

\node[noeud,draw=lblue, ultra thick] (0) at (0,0) {$B$};
\node[noeud,draw=lred, ultra thick] (1) at (1,0) {$C$};
\node[noeud,draw=lgreen, ultra thick] (2) at (2,0) {$A$};
\node[noeud,draw=lgreen, ultra thick] (3) at (3,0) {$A$};

\node[noeud,draw=lblue, ultra thick] (0b) at (0,-1) {$\beta$};
\node[noeud,draw=lred, ultra thick] (1b) at (1,-1) {$\gamma$};
\node[noeud,draw=lgreen, ultra thick] (2b) at (2,-1) {$\alpha$};
\node[noeud,draw=lgreen, ultra thick] (3b) at (3,-1) {$\alpha$};

\tikzstyle{arc}=[-,>=latex]

\draw[arc,draw=lblue, ultra thick] (0)--(0b);
\draw[arc,draw=lred, ultra thick] (1)--(1b);
\draw[arc,draw=lgreen, ultra thick] (2)--(2b);
\draw[arc,draw=lgreen, ultra thick] (3)--(3b);

\node[noeud,draw=lgreen, ultra thick] (A) at (0,-2) {$A$};
\node[noeud,draw=lblue, ultra thick] (B) at (1.5,-2) {$B$};
\node[noeud,draw=lred, ultra thick] (C) at (3,-2) {$C$};

\node[noeud,draw=lgreen, ultra thick] (a) at (0,-3) {$\alpha$};
\node[noeud,draw=lblue, ultra thick] (b) at (1.5,-3) {$\beta$};
\node[noeud,draw=lred, ultra thick] (c) at (3,-3) {$\gamma$};

\draw[arc,draw=lblue, ultra thick] (B)--(b);
\draw[arc,draw=lgreen, ultra thick] (A)--(a);
\draw[arc,draw=lred, ultra thick] (C)--(c);

\end{tikzpicture}
\end{minipage}\hfill
\begin{minipage}{0.5\textwidth}
\begin{tikzpicture}[xscale=0.7,yscale=0.7]
\tikzstyle{noeud}=[draw,circle,fill=white,scale=0.8]

\node at (-2,0) {Message 1};
\node at (-2,-1) {Message 2};
\node at (-2,-2.5) {Induced cipher};

\draw[->] (-1,0)--(4,0);
\draw[->] (-1,-1)--(4,-1);

\node[noeud,draw=lblue, ultra thick] (0) at (0,0) {$B$};
\node[noeud,draw=lred, ultra thick] (1) at (1,0) {$C$};
\node[noeud,draw=lred, ultra thick] (2) at (2,0) {$C$};
\node[noeud] (3) at (3,0) {$A$};

\node[noeud,draw=lblue, ultra thick] (0b) at (0,-1) {$\beta$};
\node[noeud,draw=lred, ultra thick] (1b) at (1,-1) {$\gamma$};
\node[noeud,draw=lred, ultra thick] (2b) at (2,-1) {$\alpha$};
\node[noeud] (3b) at (3,-1) {$\alpha$};

\tikzstyle{arc}=[-,>=latex]

\draw[arc,draw=lblue, ultra thick] (0)--(0b);
\draw[arc,draw=lred, ultra thick] (1)--(1b);
\draw[arc,draw=lred, ultra thick] (2)--(2b);

\node[noeud,draw=lgreen, ultra thick] (A) at (0,-2) {$A$};
\node[noeud,draw=lblue, ultra thick] (B) at (1.5,-2) {$B$};
\node[noeud,draw=lred, ultra thick] (C) at (3,-2) {$C$};

\node[noeud,draw=lred, ultra thick] (a) at (0,-3) {$\alpha$};
\node[noeud,draw=lblue, ultra thick] (b) at (1.5,-3) {$\beta$};
\node[noeud,draw=lred, ultra thick] (c) at (3,-3) {$\gamma$};

\draw[arc,draw=lblue, ultra thick] (B)--(b);
\draw[arc,draw=lred, ultra thick] (C)--(c);
\draw[arc,draw=lred, ultra thick] (C)--(a);

\end{tikzpicture}
\end{minipage}
\caption{\small Simple substitution cipher induced by the order of letters on two examples, one where the two messages are isomorphic  (left), and one where there are not (right). In the latter, the last letter of both messages is ignored as an inconsistency is detected at the penultimate letter.}
\label{fig:sequence_cipher}
\end{figure}

In this article, we are interested in the analogous problem of determining whether two messages are identical up to a substitution cipher, but instead of a linear sequence, the letters are placed as labels on nodes of unordered trees -- i.e. for which the order among children of a same node is not relevant. 

Instead of requiring that the two messages are of same length -- as it was the case for sequences, we require that the two trees are \emph{isomorphic}, i.e. they share the same topology. The reading order of letters is not induced by the sequence but by a \emph{tree isomorphism}, that is a bijection between the nodes of both trees, that respect topology constraints. While the reading order is unique for sequences, for trees, the number of isomorphisms is given by a product of as many factorials as the number of nodes of the tree (see upcoming equation (\ref{eq:nb_isom}) and illustrative Fig.~\ref{fig:cardequiv}). Although this number depends highly on the topology, ignoring pathological cases, it is usually extremely large.  To give an order of magnitude, for a million replicates of random recursive trees \cite{zhang2015number} of size 100, the average number of tree isomorphisms is $6.88\times 10^8$ -- with a median of $2.21 \times 10^5$. The \emph{tree ciphering isomorphism problem} can then be precised as:

\emph{``Given two isomorphic unordered trees, is there any tree isomorphism that induces a substitution cipher of the labels of one tree onto the other?''}

This question induces an equivalence relation on trees: two topologically isomorphic unordered trees with labels are equivalent if and only if there exists a tree isomorphism that induces a substitution cipher on the labels that transforms one tree onto the other -- see Theorem~\ref{th:equiv_relation}. The problem is formally introduced in this paper in Section~\ref{sec:problem},  while an example is provided now in Fig.~\ref{fig:tree_exple}.

\begin{figure}[!h]
\begin{minipage}{0.33\textwidth}
\centering
\begin{tikzpicture}[xscale=0.7,yscale=0.7]
\tikzstyle{noeud}=[draw,circle,fill=white,scale=0.8]
\tikzstyle{arc}=[-,>=latex]

\node[noeud] (u1) at (0,0) {$A$};
\node[noeud] (u2) at (-1,-1) {$B$};
\node[noeud] (u3) at (1,-1) {$B$};
\node[noeud] (u4) at (-1,-2) {$A$};
\node[noeud] (u5) at (1,-2) {$C$};

\draw[arc] (u1)--(u2);
\draw[arc] (u1)--(u3);
\draw[arc] (u2)--(u4);
\draw[arc] (u3)--(u5);

\node[noeud] (v1) at (3,0) {$\alpha$};
\node[noeud] (v2) at (2,-1) {$\beta$};
\node[noeud] (v3) at (4,-1) {$\beta$};
\node[noeud] (v4) at (2,-2) {$\alpha$};
\node[noeud] (v5) at (4,-2) {$\gamma$};

\draw[arc] (v1)--(v2);
\draw[arc] (v1)--(v3);
\draw[arc] (v2)--(v4);
\draw[arc] (v3)--(v5);

\node at (0,0.6) {$T_1$};
\node at (3,0.6) {$T_2$};

\end{tikzpicture}
\end{minipage}\hfill
\begin{minipage}{0.64\textwidth}
\caption{\small Two messages encoded as labels on unordered trees $T_1$ and $T_2$ (left).  $T_1$ and $T_2$ are topologically identical. There exist two tree isomorphisms between $T_1$ and $T_2$, one inducing a simple substitution cipher (below, left) and the other one that does not (below, right). In the latter, the full tree isomorphism is not parsed as an inconsistency is detected before. Overall, the two labeled trees $T_1$ and $T_2$ are isomorphic since at least one tree isomorphism leads to a substitution cipher.}
\label{fig:tree_exple}
\end{minipage}

\begin{minipage}{0.5\textwidth}
\centering
\begin{tikzpicture}[xscale=0.7,yscale=0.7]
\tikzstyle{noeud}=[draw,circle,fill=white,scale=0.8]
\tikzstyle{arc}=[-,>=latex]

\node[noeud,draw=lgreen, ultra thick] (u1) at (0,0) {$A$};
\node[noeud,draw=lblue, ultra thick] (u2) at (-1,-1) {$B$};
\node[noeud,draw=lblue, ultra thick] (u3) at (1,-1) {$B$};
\node[noeud,draw=lgreen, ultra thick] (u4) at (-1,-2) {$A$};
\node[noeud,draw=lred, ultra thick] (u5) at (1,-2) {$C$};

\draw[arc] (u1)--(u2);
\draw[arc] (u1)--(u3);
\draw[arc] (u2)--(u4);
\draw[arc] (u3)--(u5);

\node[noeud,draw=lgreen, ultra thick] (v1) at (4,0) {$\alpha$};
\node[noeud,draw=lblue, ultra thick] (v2) at (3,-1) {$\beta$};
\node[noeud,draw=lblue, ultra thick] (v3) at (5,-1) {$\beta$};
\node[noeud,draw=lgreen, ultra thick] (v4) at (3,-2) {$\alpha$};
\node[noeud,draw=lred, ultra thick] (v5) at (5,-2) {$\gamma$};

\draw[arc] (v1)--(v2);
\draw[arc] (v1)--(v3);
\draw[arc] (v2)--(v4);
\draw[arc] (v3)--(v5);

\node at (0,0.6) {$T_1$};
\node at (4,0.6) {$T_2$};

\draw[arc,draw=lgreen, ultra thick] (u1)--(v1);
\draw[arc,draw=lgreen, ultra thick] (u4) to[bend left=-45](v4);
\draw[arc,draw=lred, ultra thick] (u5) to[bend left=-45](v5);
\draw[arc,draw=lblue, ultra thick] (u2) to[bend left=30](v2);
\draw[arc,draw=lblue, ultra thick] (u3) to[bend left=30](v3);

\node at (2,-3.4) {Induced cipher};

\node[noeud,draw=lgreen, ultra thick] (A) at (0.5,-4) {$A$};
\node[noeud,draw=lblue, ultra thick] (B) at (2,-4) {$B$};
\node[noeud,draw=lred, ultra thick] (C) at (3.5,-4) {$C$};

\node[noeud,draw=lgreen, ultra thick] (a) at (0.5,-5) {$\alpha$};
\node[noeud,draw=lblue, ultra thick] (b) at (2,-5) {$\beta$};
\node[noeud,draw=lred, ultra thick] (c) at (3.5,-5) {$\gamma$};

\draw[arc,draw=lblue, ultra thick] (B)--(b);
\draw[arc,draw=lgreen, ultra thick] (A)--(a);
\draw[arc,draw=lred, ultra thick] (C)--(c);

\end{tikzpicture}
\end{minipage}~
\begin{minipage}{0.5\textwidth}
\centering
\begin{tikzpicture}[xscale=0.7,yscale=0.7]
\tikzstyle{noeud}=[draw,circle,fill=white,scale=0.8]
\tikzstyle{arc}=[-,>=latex]

\node[noeud,draw=lgreen, ultra thick] (u1) at (0,0) {$A$};
\node[noeud,draw=lblue, ultra thick] (u2) at (-1,-1) {$B$};
\node[noeud,draw=lblue, ultra thick] (u3) at (1,-1) {$B$};
\node[noeud,draw=lgreen, ultra thick] (u4) at (-1,-2) {$A$};
\node[noeud] (u5) at (1,-2) {$C$};

\draw[arc] (u1)--(u2);
\draw[arc] (u1)--(u3);
\draw[arc] (u2)--(u4);
\draw[arc] (u3)--(u5);

\node[noeud,draw=lgreen, ultra thick] (v1) at (4,0) {$\alpha$};
\node[noeud,draw=lblue, ultra thick] (v2) at (3,-1) {$\beta$};
\node[noeud,draw=lblue, ultra thick] (v3) at (5,-1) {$\beta$};
\node[noeud] (v4) at (3,-2) {$\alpha$};
\node[noeud,draw=lgreen, ultra thick] (v5) at (5,-2) {$\gamma$};

\draw[arc] (v1)--(v2);
\draw[arc] (v1)--(v3);
\draw[arc] (v2)--(v4);
\draw[arc] (v3)--(v5);

\node at (0,0.6) {$T_1$};
\node at (4,0.6) {$T_2$};

\draw[arc,draw=lgreen, ultra thick] (u1)--(v1);
\draw[arc,draw=lgreen, ultra thick] (u4)to[bend left=-30](v5);
\draw[arc,draw=lblue, ultra thick] (u2) to[bend left=20](v3);
\draw[arc,draw=lblue, ultra thick] (u3)--(v2);

\node at (2,-3.4) {Induced cipher};

\node[noeud,draw=lgreen, ultra thick] (A) at (0.5,-4) {$A$};
\node[noeud,draw=lblue, ultra thick] (B) at (2,-4) {$B$};
\node[noeud,draw=lred, ultra thick] (C) at (3.5,-4) {$C$};

\node[noeud,draw=lgreen, ultra thick] (a) at (0.5,-5) {$\alpha$};
\node[noeud,draw=lblue, ultra thick] (b) at (2,-5) {$\beta$};
\node[noeud,draw=lgreen, ultra thick] (c) at (3.5,-5) {$\gamma$};

\draw[arc,draw=lblue, ultra thick] (B)--(b);
\draw[arc,draw=lgreen, ultra thick] (A)--(a);
\draw[arc,draw=lgreen, ultra thick] (A)--(c);

\end{tikzpicture}
\end{minipage}
\end{figure}

Determining if two trees are \emph{topologically} isomorphic can be achieved within linear time via the so-called AHU algorithm \cite[Ex.~3.2]{aho1974design}. Determining if two labeled trees are isomorphic under the definition above is, on the other hand, a difficult problem. It is an instance of labeled graph isomorphism -- see \cite{zemlyachenko1985graph} and \cite{champin2003measuring} -- that was introduced under the name \emph{marked tree isomorphism} in \cite[Section~6.4]{booth1979problems}, where it has been proved \emph{graph isomorphism complete}, i.e. as hard as graph isomorphism. The latter is still an open problem, where no proof of NP-completeness nor polynomial algorithm is known \cite{schoning1987graph}.

One classic family of algorithms trying to achieve graph isomorphism are \emph{color refinement algorithms}, also known as Weisfeiler-Leman algorithms \cite{weisfeiler1968reduction}. Both graphs are colored according to some rules, and the color histograms are compared afterwards : if they diverge, the graphs are not isomorphic. However, this test is incomplete in the sense that there exist non-isomorphic graphs that are not distinguished by the coloring. The distinguishability of those algorithms is constantly improved -- see \cite{grohe2021deep} for recent results -- but does not yet answer the problem for any graph. Actually, AHU algorithm for topological tree isomorphism can be interpreted as a color refinement algorithm.

To address the tree ciphering isomorphism problem, one strategy is to explore the space of tree isomorphisms and look for one that induces a ciphering, if it exists. As stated earlier and as discussed in Section~\ref{sec:problem}, such a search space is factorially large. This paper does not seek to solve the tree ciphering isomorphism problem, but rather to break the combinatorial complexity of the search space.

In Section~\ref{sec:break}, we present an algorithm fulfilling this objective. Even if it uses AHU algorithm, our method does not involve a color refinement process. Actually, we adopt a strategy that is more related to constrained matching problems in bipartite graphs \cite{canzar2015tree,mastrolilli2012constrained}. In details, since we are building two isomorphisms simultaneously -- one on trees and the other on labels -- that must be compatible, the general idea is to use the constraints of one to make deductions about the other, and vice versa. For instance, whenever two nodes must be mapped together, so are their labels, and therefore you can eliminate all potential tree isomorphisms that would have mapped those labels differently. When no more deductions are possible, our algorithm stops. To complete (if feasible) the two isomorphisms, and to explore the remaining space, different strategies can be considered, including, for example, backtracking. However, this is not the purpose of this paper which aims to break the combinatorial complexity of the space of tree isomorphisms compatible to substitution ciphering.

Finally, in Section~\ref{sec:analysis}, we show that our algorithm runs in linear time -- at least experimentally. Moreover, we show on simulated data that it reduces on average the cardinality of the search space of an \emph{exponential factor} -- which shows the great interest of this approach especially considering its low computational cost.

\section{Problem formulation}\label{sec:problem}
\subsection{Tree isomorphisms}\label{ss:background}
A (rooted) tree is a connected directed graph without cycle such that (i) there exists a special node called the root, which has no parent, and (ii) any node different from the root has exactly one parent. The parent of a node $u$ is denoted by $\parent(u)$, where its children are denoted as $\children(u)$. Trees are said to be unordered if the order among siblings is not significant. In a sequel, we use \emph{tree} to designate a unordered rooted tree.

The \emph{degree} of a node is defined as $\deg(u)=\#\children(u)$, and the degree of a tree is $\deg(T) = \max_{u\in T} \deg(u)$. The leaves $\leaves(T)$ of a tree $T$ are all the nodes without any children. The depth $\depth(u)$ of a node $u$ is the length of the path between $u$ and the root. The depth $\depth(T)$ of $T$ is the maximal depth among all nodes. For any node $u$ of $T$, we define the \emph{subtree} $T[u]$ rooted in $u$ as the tree composed of $u$ and all of its descendants.

Let $T_1$ and $T_2$ be two trees.
\begin{definition}\label{def:isomtree}
A bijection $\varphi : T_1 \to T_2$ is a tree isomorphism if and only if, for any $u,v\in T_1$, if $u$ is a child of $v$ in $T_1$, then $\varphi(u)$ is a child of $\varphi(v)$ in $T_2$; in addition, roots must be mapped together.
\end{definition}
We can define $\isom(T_1,T_2)$ as the set of all tree isomorphisms between $T_1$ and $T_2$. If this set is not empty, then $T_1$ and $T_2$ are \emph{topologically isomorphic} and we denote $T_1\equiv T_2$. It is well known that $\equiv$ is an equivalence relation over the set of trees  \cite[Chapter~4]{valiente2002algorithms}. Fig.~\ref{fig:isomtree} provides an example of tree isomorphism.

\begin{figure}[H]
\centering
\def\xscale{0.25}
\def\yscale{0.5}
\def\nodescale{0.7}
\def\position{1pt}
\begin{minipage}{0.46\textwidth}
\centering
\begin{tikzpicture}[xscale=\xscale,yscale=\yscale]
\tikzstyle{noeud}=[draw,circle,fill=white,scale=\nodescale*1]
\tikzstyle{attribut}=[scale=\nodescale*1,font=\bf]
\tikzstyle{arc}=[->-,>=latex]
\tikzstyle{fleche}=[->,>=latex,red,thick]

\node[noeud,fill=lred] (u1) at (2,2) {a};
\node[noeud,fill=lgreen] (u2) at (6,0) {b};
\node[noeud,fill=lblue] (u3) at (2,0) {c};
\node[noeud,fill=lpurple] (u4) at (5,-2) {d};
\node[noeud,fill=lorange] (u5) at (7,-2) {e};
\node[noeud,fill=lpink] (u6) at (1,-2) {f};
\node[noeud,fill=lbrown] (u7) at (3,-2) {g};
\node[noeud,fill=lyellow] (u8) at (-2,0) {h};

\node (t1) at (2,3) {$T_1$};

\draw[arc] (u1)--(u2) ;
\draw[arc] (u1)--(u3) ;
\draw[arc] (u1)--(u8) ;
\draw[arc] (u2)--(u4) ;
\draw[arc] (u2)--(u5) ;
\draw[arc] (u3)--(u6) ;
\draw[arc] (u3)--(u7) ;

\def\xshift{12}

\node[noeud,fill=lred] (v1) at (2+\xshift,2) {1};
\node[noeud,fill=lblue] (v2) at (-2+\xshift,0) {2};
\node[noeud,fill=lgreen] (v3) at (2+\xshift,0) {3};
\node[noeud,fill=lbrown] (v4) at (-3+\xshift,-2) {4};
\node[noeud,fill=lpink] (v5) at (-1+\xshift,-2) {5};
\node[noeud,fill=lpurple] (v6) at (1+\xshift,-2) {6};
\node[noeud,fill=lorange] (v7) at (3+\xshift,-2) {7};
\node[noeud,fill=lyellow] (v8) at (6+\xshift,0) {8};

\node (t1) at (2+\xshift,3) {$T_2$};

\draw[arc] (v1)--(v2) ;
\draw[arc] (v1)--(v3) ;
\draw[arc] (v1)--(v8) ;
\draw[arc] (v2)--(v4) ;
\draw[arc] (v2)--(v5) ;
\draw[arc] (v3)--(v6) ;
\draw[arc] (v3)--(v7) ;

\end{tikzpicture}

\end{minipage}\hfill
\begin{minipage}{0.5\textwidth}
\centering
\setlength\tabcolsep{3pt} 
\begin{tabular}{r|cccccccc}
$u \in T_1$ & \colorbox{lred}{a} &\colorbox{lgreen}{b} &\colorbox{lblue}{c}&\colorbox{lpurple}{d} &\colorbox{lorange}{e} &\colorbox{lpink}{f} &\colorbox{lbrown}{g} &\colorbox{lyellow}{h} \\
\hline
$\varphi(u)\in T_2$ &  \colorbox{lred}{1} &\colorbox{lgreen}{3}& \colorbox{lblue}{2}& \colorbox{lpurple}{6}& \colorbox{lorange}{7}& \colorbox{lpink}{5}& \colorbox{lbrown}{4}& \colorbox{lyellow}{8}
\end{tabular}
\caption{\small Two topologically isomorphic trees $T_1$ and $T_2$ (left) and an example of tree isomorphism $\varphi\in \isom(T_1,T_2)$ (above). Nodes are labeled and colored for ease of comprehension.}
\label{fig:isomtree}
\end{minipage}
\end{figure}
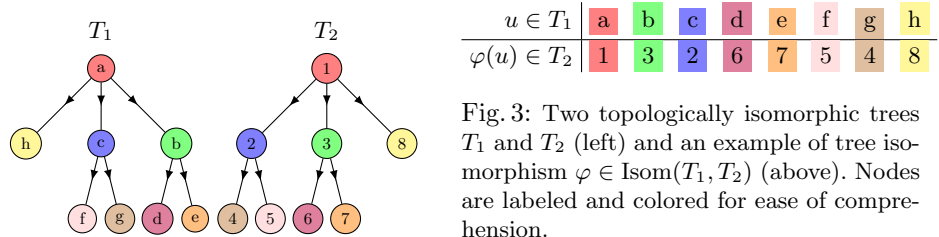

The class of equivalence of node $u\in T_i$ under $\equiv$ -- denoted by $[u]$ -- is the set of all nodes $v\in T_i$ such that $T_i[u]\equiv T_i[v]$.  So-called AHU algorithm \cite[Ex.~3.2]{aho1974design} assigns in a bottom-up manner to each node $u$ of both trees a color that represents $[u]$. The algorithm can thereby conclude in linear time whether two trees are isomorphic, if and only if their roots are identically colored.

\begin{wrapfigure}[6]{R}{0.5\textwidth}
\vspace{-1.5\baselineskip}
\centering
\begin{minipage}{0.25\textwidth}
\centering
\def\xscale{0.25}
\def\yscale{0.4}
\def\nodescale{0.7}

\begin{tikzpicture}[xscale=\xscale,yscale=\yscale]
\tikzstyle{noeud}=[draw,circle,fill=white,scale=\nodescale*1]
\tikzstyle{attribut}=[scale=\nodescale*1,font=\bf]
\tikzstyle{arc}=[->-,>=latex]
\tikzstyle{fleche}=[->,>=latex,red,thick]

\node[noeud,fill=lblue] (u1) at (2,2) {};
\node[noeud,fill=lred] (u2) at (6,0) {};
\node[noeud,fill=lred] (u3) at (2,0) {};
\node[noeud,fill=lgreen] (u4) at (5,-2) {};
\node[noeud,fill=lgreen] (u5) at (7,-2) {};
\node[noeud,fill=lgreen] (u6) at (1,-2) {};
\node[noeud,fill=lgreen] (u7) at (3,-2) {};
\node[noeud,fill=lgreen] (u8) at (-2,0) {};

\draw[arc] (u1)--(u2) ;
\draw[arc] (u1)--(u3) ;
\draw[arc] (u1)--(u8) ;
\draw[arc] (u2)--(u4) ;
\draw[arc] (u2)--(u5) ;
\draw[arc] (u3)--(u6) ;
\draw[arc] (u3)--(u7) ;

\draw[rounded corners,thick,red] (0,-2.5) rectangle (3.9,-1.5);
\draw[rounded corners,thick,red] (4.1,-2.5) rectangle (8,-1.5);
\draw[rounded corners,thick,red] (1,-0.5) rectangle (7,0.5);

\node[red] at (-0.5,-1.5) {$2$};
\node[red] at (0.5,0.5) {$2$};
\node[red] at (8.5,-1.5) {$2$};

\end{tikzpicture}
\end{minipage}~
\begin{minipage}{0.25\textwidth}
\caption{\small A tree $T$. Nodes susceptible to be swapped are boxed together, leading to $N_\equiv(T) = (2!)^3 = 8$.}
\label{fig:cardequiv}
\end{minipage}
\end{wrapfigure}

Any tree isomorphism $\varphi : T_1\to T_2$ maps $u\in T_1$ onto $v=\varphi(u)\in T_2$ only if $[u]=[v]$. Thus, all tree isomorphisms can be -- recursively from the root -- obtained by swapping nodes (i) of same equivalence class and (ii) children of a same node. Consequently, the number of tree isomorphisms between $T_1$ and $T_2$ depends only on the class of equivalence of $T_1$ (equivalently $T_2$), and will be denoted by $N_\equiv(T_1)$. For any tree $T$, we have 
\begin{equation}\label{eq:nb_isom}
N_\equiv(T)= \prod_{u\in T}\prod_{q \in \lbrace [v] : v \in \children(u)\rbrace} \left(\# \lbrace v \in \children(u) :[v]=q\rbrace\right)!.
\end{equation}
An example is provided in Fig.~\ref{fig:cardequiv}.

\subsection{Tree cipherings}\label{ss:cipherings}
We now assume that each node of a tree carries a \emph{label}. Let $T$ be a tree and $u\in T$; we denote by $\overline{u}$ the label of node $u$. The alphabet of $T$, denoted by $\attr(T)$, is defined as $\attr(T) = \cup_{u\in T} \overline{u}$. We say that $T$ is a labeled tree.

Let $T_1$ and $T_2$ be two topologically isomorphic labeled trees and $\varphi\in\isom(T_1,T_2)$. $\varphi$ naturally induces a binary relation $\mathrel{R_\varphi}$ over sets $\attr(T_1)$ and $\attr(T_2)$, defined as
\[
\forall x\in \attr(T_1), \forall y\in\attr(T_2),  x \mathrel{R_\varphi} y \iff \exists u\in T_1, (x=\overline{u} )\wedge (y=\overline{\varphi(u)}).
\]
Fig.~\ref{fig:rel_induite} illustrates this induced binary relation on an example.

\begin{figure}[h]
\centering
\def\xscale{0.25}
\def\yscale{0.5}
\def\nodescale{0.7}
\begin{minipage}{0.45\textwidth}
\centering
\begin{tikzpicture}[xscale=\xscale,yscale=\yscale]
\tikzstyle{noeud}=[draw,circle,fill=white,scale=\nodescale*1]
\tikzstyle{attribut}=[scale=\nodescale*1,font=\bf]
\tikzstyle{arc}=[->-,>=latex]
\tikzstyle{fleche}=[->,>=latex,red,thick]

\node[noeud,fill=lred] (u1) at (2,2) {$A$};
\node[noeud,fill=lgreen] (u2) at (6,0) {$B$};
\node[noeud,fill=lblue] (u3) at (2,0) {$C$};
\node[noeud,fill=lpurple] (u4) at (5,-2) {$C$};
\node[noeud,fill=lorange] (u5) at (7,-2) {$B$};
\node[noeud,fill=lpink] (u6) at (1,-2) {$A$};
\node[noeud,fill=lbrown] (u7) at (3,-2) {$C$};
\node[noeud,fill=lyellow] (u8) at (-2,0) {$D$};

\node (t1) at (2,3) {$T_1$};

\draw[arc] (u1)--(u2) ;
\draw[arc] (u1)--(u3) ;
\draw[arc] (u1)--(u8) ;
\draw[arc] (u2)--(u4) ;
\draw[arc] (u2)--(u5) ;
\draw[arc] (u3)--(u6) ;
\draw[arc] (u3)--(u7) ;

\def\xshift{12}

\node[noeud,fill=lred] (v1) at (2+\xshift,2) {$\alpha$};
\node[noeud,fill=lblue] (v2) at (-2+\xshift,0) {$\gamma$};
\node[noeud,fill=lgreen] (v3) at (2+\xshift,0) {$\alpha$};
\node[noeud,fill=lbrown] (v4) at (-3+\xshift,-2) {$\gamma$};
\node[noeud,fill=lpink] (v5) at (-1+\xshift,-2) {$\alpha$};
\node[noeud,fill=lpurple] (v6) at (1+\xshift,-2) {$\gamma$};
\node[noeud,fill=lorange] (v7) at (3+\xshift,-2) {$\beta$};
\node[noeud,fill=lyellow] (v8) at (6+\xshift,0) {$\gamma$};

\node (t1) at (2+\xshift,3) {$T_2$};

\draw[arc] (v1)--(v2) ;
\draw[arc] (v1)--(v3) ;
\draw[arc] (v1)--(v8) ;
\draw[arc] (v2)--(v4) ;
\draw[arc] (v2)--(v5) ;
\draw[arc] (v3)--(v6) ;
\draw[arc] (v3)--(v7) ;

\end{tikzpicture}
\end{minipage}~
\begin{minipage}{0.2\textwidth}
\centering
\begin{tikzpicture}[xscale=\xscale,yscale=\yscale]
\tikzstyle{noeud}=[draw,circle,fill=white,scale=\nodescale*1]
\tikzstyle{arc}=[->,>=latex]

\node[noeud] (a) at (0,4.5) {$A$};
\node[noeud] (b) at (0,3) {$B$};
\node[noeud] (c) at (0,1.5) {$C$};
\node[noeud] (d) at (0,0) {$D$};

\node[noeud] (1) at (5,3.75) {$\alpha$};
\node[noeud] (2) at (5,2.25) {$\beta$};
\node[noeud] (3) at (5,0.75) {$\gamma$};

\node (a1) at (0,6) {$\attr(T_1)$};
\node (a1) at (5,6) {$\attr(T_2)$};

\draw[arc] (a)--(1) ;
\draw[arc] (b)--(1) ;
\draw[arc] (b)--(2) ;
\draw[arc] (c)--(3) ;
\draw[arc] (d)--(3) ;

\end{tikzpicture}
\end{minipage}\hfill
\begin{minipage}{0.3\textwidth}
\centering
\caption{\small Two topologically isomorphic labeled trees (left) and the induced binary relation (right). The tree isomorphism $\varphi$ is displayed through node colors -- cf. Fig.~\ref{fig:isomtree}.}
\label{fig:rel_induite}
\end{minipage}
\end{figure}

Such a relation $\mathrel{R_\varphi}$ is said to be a bijection if and only if for any $x\in \attr(T_1)$, there exists a unique $y\in \attr(T_2)$ so that $x\mathrel{R_\varphi} y$, and conversely if for any $y\in \attr(T_2)$, there exists a unique $x\in \attr(T_1)$ so that $x\mathrel{R_\varphi} y$. This is not the case of the relation induced by the example in Fig.~\ref{fig:rel_induite}, since $C$ and $D$ are both in relation to $\gamma$, and also $B$ is in relation to both $\alpha$ and $\beta$.

When $\mathrel{R_\varphi}$ is a bijection, we can define a bijective function $f_\varphi : \attr(T_1) \to \attr(T_2)$ by $f_\varphi(x)=y\iff x\mathrel{R_\varphi} y$. This function is called a \emph{substitution cipher} (following the analogy developed in the introduction) and verifies $\forall u\in T_1, f_\varphi(\overline{u}) = \overline{\varphi(u)}.$

\begin{definition}\label{def:cipher-tree-isom}
$\varphi \in \isom(T_1,T_2)$ is said to be a \emph{tree ciphering} if and only if $\mathrel{R_\varphi}$ is a bijection; in which case we denote $T_1\tf{\varphi} T_2$.
\end{definition}
Let us denote by $\enc(T_1,T_2)$ the set of tree cipherings between $T_1$ and $T_2$. If $\enc(T_1,T_2)$ is not empty, then we write $T_1\sim T_2$ and say that $T_1$ and $T_2$ are \emph{isomorphic by substitution ciphering}, since the following results holds.
\begin{theorem}\label{th:equiv_relation}
  $\sim$ is an equivalence relation over the set of labeled trees.
\end{theorem}
\begin{proof}
The proof is deferred to Appendix~\ref{annex:proof}.
\end{proof}


\begin{remark}\label{rmk:restriction}
It is possible to be more restrictive on the choices of substitution ciphers. Let $(G,\circ)$ be a subgroup of the bijections between $\attr(T_1)$ and $\attr(T_2)$. Then, if we replace ``$\mathrel{R_\varphi}$ is a bijection'' in Definition~\ref{def:cipher-tree-isom} by ``$\mathrel{R_\varphi} \in G$'', the induced relation $\sim_G$ is also an equivalence relation.
With $G=\lbrace \id \rbrace$, $T_1\sim_G T_2$ means $T_1\equiv T_2$ plus equality of labels. It is actually the definition adopted for labeled tree isomorphism in \cite[Section~5.1]{azais2020weight}. 
\end{remark}

Determining if $T_1\sim T_2$ implies to find $\varphi\in\isom(T_1,T_2)$ such that $\varphi$ is also in $\enc(T_1,T_2)$. Therefore, the cardinality of the search space is given by (\ref{eq:nb_isom}), and is potentially exponentially large compared to the size of the trees. In the sequel of the paper, we present an algorithm that aims to break this cardinality.

\section{Breaking down the combinatorial complexity}\label{sec:break}
Let be two labeled trees $T_1$ and $T_2$. To build a tree ciphering between $T_1$ and $T_2$ (if only it exists), a strategy is to ensure that $T_1\equiv T_2$, and then explore $\isom(T_1,T_2)$, whose cardinality is given by (\ref{eq:nb_isom}). Since AHU algorithm \cite[Ex.~3.2]{aho1974design} solves the problem of determining whether $T_1\equiv T_2$ in linear time, as well as assigning to each node $u$ its equivalence class $[u]$ under $\equiv$, we use AHU as a preprocessing step.

In the case of linear messages, illustrated in Fig.~\ref{fig:sequence_cipher}, the isomorphism on labels is induced by the reading order, starting with the first letter. In our case, we know that the roots have to be mapped together and we start here. At each step of the algorithm, we will add elements to the two bijections we aim to build: $\varphi$ for the nodes and $f$ for the labels. We present in Subsection~\ref{ss:extension} how to update those bijections, with the \textsc{ExtBij} procedure.

Besides, the topological constraints imposed by tree isomorphism allow to sort the nodes of the trees and to group them by susceptibility to be mapped together. In Subsection~\ref{ss:bags}, we introduce two concepts, bags and collections, that reflects this grouping mechanism. The actual mapping of nodes is performed by the procedure \textsc{MapNodes}, introduced in Subsection~\ref{ss:map}.

Finally, the precise course of the algorithm is presented in Subsection~\ref{ss:preprocessing}. Starting by grouping all the nodes together, we successively add topological filters to refine the groups of nodes. Whenever possible, if a filter allows us to deduce that two nodes should be mapped together, we do so, thus reducing the cardinality of the remaining possibilities. The last filter checks constraints on labels and allows a last phase of deductions, before concluding the algorithm -- whose analysis is discussed in Section~\ref{sec:analysis}.

The course of the algorithm is illustrated through an example in Appendix~\ref{annex:example}.

\subsection{Extension of a bijection}\label{ss:extension}
During the execution of the algorithm, we construct two mappings: $\varphi$ for the nodes, and $f$ for the labels. They start as empty mappings $\emptyset \mapsto \emptyset$, and will be updated through time. They must remain bijective at all times, and the rules for updating them are presented here.

A partial bijection $\psi$ from $E$ to $F$ is an injective function from a subset $S_\psi$ of $E$ to $F$. Let $a\in E$ and $b\in F$; suppose we want to determine if the couple $(a,b)$ is compatible with $\psi$ -- in the sense that it respects (or does not contradict) the partial bijection. First, if $a\in S_\psi$, then $b$ must be equal to $\psi(a)$. Otherwise, if $a\not\in S_\psi$, then $b$ must not be in the image of $\psi$, i.e. $\forall s\in S_\psi, \psi(s)\neq b$. If those conditions are respected, then $(a,b)$ is compatible with $\psi$; furthermore, if $a\not\in S_\psi$, then we can extend $\psi$ on $S_\psi\cup \lbrace a\rbrace$ by defining $\psi(a)=b$ so that $\psi$ remains a partial bijection. Formally, for any $a\in E$ and $b\in F$, with $\psi$ a partial bijection from $E$ to $F$, we define
\[\textsc{ExtBij}(a,b,\psi) = \big( a\in S_\psi \implies \psi(a)=b\big) \wedge \big(a\not\in S_\psi \implies \forall s\in S_\psi, \psi(s)\neq b\big);\]
so that \textsc{ExtBij}$(a,b,\psi)$ returns $\top$ if and only if the couple $(a,b)$ is compatible with the partial bijection $\psi$. For the sake of brevity, we assume that the function \textsc{ExtBij} also extends the partial bijection in the case $a\not\in S_\psi$ by defining $\psi(a)=b$ -- naturally only if the function returned $\top$.

\textsc{ExtBij} will be used in the sequel to update both partial bijections $\varphi$ (from $T_1$ to $T_2$) and $f$ (from $\attr(T_1)$ to $\attr(T_2)$). However, if one uses the restricted substitution ciphers presented in Remark~\ref{rmk:restriction}, one must design a specific version of \textsc{ExtBij} to update $f$, accounting for the desired properties.

\subsection{Bags and collections}\label{ss:bags}
Remark that if two nodes $u\in T_1$ and $v\in T_2$ are mapped together via $\varphi$, then they must share a number of common features: (i) $\depth(u)=\depth(v)$, (ii) $[u]=[v]$, (iii) $f_\varphi(\overline{u})=\overline{v}$, and (iv) $\varphi(\parent(u))=\parent(v)$. Our goal is to gather together nodes that share such common features. For this purpose, we introduce the concepts of bags and collections. 

We recall that a partition $P$ of a set $X$ is a set of non-empty subsets $P_i$ of $X$ such that every element $x\in X$ is in exactly one of these subsets $P_i$.  Let $P$ (resp. $Q$) be a partition of the nodes of $T_1$ (resp. $T_2$).

A bag $B$ is a couple $(P_i, Q_j)$ such that $P_i\in P, Q_j\in Q$ and $\# P_i = \# Q_j$ -- this number is denoted by $\#B$. A bag contains nodes that share a number of common features, and are therefore candidates to be mapped together. If a bag is constructed such that $P_i$ and $Q_j$ each contain a single element, then those elements should be unambiguously mapped together  -- via the function \textsc{MapNodes} that will be introduced in the next subsection. Formally, this rule is expressed as:
 
 \begin{deducrule}\label{rule:mapping}
While there exist bags $B = (P_i,Q_j)$ with $P_i=\lbrace u \rbrace$ and $Q_j = \lbrace v \rbrace$, call \textsc{MapNodes}$(u,v,\varphi,f)$ -- and delete $B$.
\end{deducrule}
 
A collection $C$ gathers several $P_i$'s and $Q_j$'s, that are candidates to form bags. Formally, $C : \mathbb{N}  \to 2^{P} \times 2^{Q}$ with
$C(n)= (\lbrace P_i, i\in I\rbrace, \lbrace Q_j, j\in J\rbrace)$ -- possibly $I=J=\emptyset$ --
such that, denoting the components by $C_1(n)$ and $C_2(n)$,          
\begin{enumerate}[(i)]
\item $\forall n, \#C_1(n) = \#C_2(n)$;
\item $\forall n, \forall P_i\in C_1(n), \#P_i=n$ and $\exists a\in \attr(T_1), \forall u\in P_i, \overline{u}=a$;
\item $\forall n, \forall Q_j\in C_2(n), \#Q_j=n$ and $\exists b\in \attr(T_2), \forall v\in Q_j, \overline{v}=b$.
\end{enumerate}
We denote by $\#C(n)$ the common cardinality of (i); and $\overline{P_i}$ and $\overline{Q_j}$ the common labels of (ii) and (iii). Note that the number of $n$'s such that $\#C(n)>0$ is finite.

The elements of $C_i(n)$, since they share the same cardinality $n$, are candidates to form bags together. If $\#C(n)=1$, we can form a bag with the two elements of $C_1(n)$ and $C_2(n)$:

\begin{deducrule}\label{rule:collection}
While there exist collections $C$ and integers $n$ for which $C(n) =( \lbrace P_i \rbrace , \lbrace Q_j \rbrace)$; if \textsc{ExtBij}$(\overline{P_i},\overline{Q_j},f)$, create bag $(P_i,Q_j)$ and delete $C(n)$  -- otherwise stop and conclude that $T_1\not\sim T_2$.
\end{deducrule}

As it will be described later on, each subset $P_i$ or $Q_j$ will belong to either a bag or a collection. Any node $u$ will either be already mapped in $\varphi$, or attached to one bag or collection through the partitions.We denote by $p(u)$ the function that returns the bag or collection in which $u$ belongs to, if any. We denote by $\bags$ the set of all bags, and by $\collections$ the set of all collections.

\subsection{Mapping Nodes}\label{ss:map}

\begin{wrapfigure}[11]{R}{0.6\textwidth}
\small
\vspace{-1\baselineskip}
\begin{algorithm}[H]
\caption{\textsc{MapNodes}}\label{algo:mapnodes}
\SetKw{KwAnd}{and}
\KwIn{$u\in T_1,v\in T_2,\varphi,f$}
\eIf{\textsc{ExtBij}$(\overline{u},\overline{v},f)$ \KwAnd \textsc{ExtBij}$(u,v,\varphi)$}{
Delete $u$ from $p(u)$ and $v$ from $p(v)$\\
\textsc{SplitChildren}$(u,v)$\\
\eIf{$(\parent(u),\parent(v))\not\in \mathcal{G}_\varphi$}{
Return \textsc{MapNodes}($\parent(u),\parent(v),\varphi,f$)}{Return $\top$}
}{
Return $\bot$
}
\end{algorithm}
\end{wrapfigure}

We now present with Algorithm~\ref{algo:mapnodes} the function \textsc{MapNodes} that performs the mapping between nodes, while updating $\varphi, f, \bags$ and $\collections$. The latter two, $\bags$ and $\collections$, are considered to be ``global'' variables and are therefore not included in the pseudocode provided. 

Once two nodes $u$ and $v$ are mapped, the topology constraints impose that $\parent(u)$ and $\parent(v)$ are mapped together, if not already the case, but also $\children(u)$ and $\children(v)$. These children are either (i) already mapped -- and there is nothing to do, or (ii) in bags or collections potentially containing other nodes with which they can no longer be mapped -- since their parents are not. In the latter case, it is then necessary to separate the children of $u$ and $v$ from these bags and collections. The procedure \textsc{SplitChildren} aims to do that, in the following manner. For each $P_i$ (resp. $Q_j$) in the current partitions of nodes such that $P_u = P_i \cap \children(u)\neq \emptyset$ (resp. $Q_v= Q_j\cap \children(v)\neq\emptyset$):

\begin{itemize}
\item Either $(P_i,Q_j)$ forms a bag, in which case we delete it and create instead two new bags formed by $(P_u,Q_v)$ and $(P_i\setminus P_u, Q_j\setminus Q_v)$.
\item Either there exists a collection $C$ so that $P_i\in C_1(n)$ and $Q_j\in C_2(n)$ -- with $n=\#P_i=\#Q_j$. In which case, we remove them from their set $C_i(n)$, and add instead $P_u$ (resp. $Q_v$) to $C_1(q)$ (resp. $C_2(q)$) -- with $q=\#P_u=\#Q_v$ -- and $P_i\setminus P_u$ (resp. $Q_j\setminus Q_v$) to $C_1(n-q)$ (resp. $C_2(n-q)$). Note that this splitting operation changes the sets $C_i(\cdot)$ and therefore we need to apply Deduction Rule~\ref{rule:collection} to check whether some bags are to be created or not.
\end{itemize}

At any time, if \textsc{MapNodes} returns $\bot$, then we can immediately conclude that $T_1\not\sim T_2$ and stop. Similarly, if the procedure \textsc{SplitChildren} leads to the creation of a pathological object (e.g. a bag where $\#P_i\neq \#Q_j$), we can also conclude that $T_1\not\not\sim T_2$ and stop. We can conclude that $T_1\sim T_2$ only when all nodes have been mapped.

\subsection{The algorithm}\label{ss:preprocessing}
Let $T_1$ and $T_2$ be two labeled trees; we assume that $T_1\equiv T_2$. Let $\varphi:\emptyset \mapsto \emptyset$ and $f : \emptyset \mapsto \emptyset$. We start with no collections and a single bag containing all nodes of $T_1$ and $T_2$. The general idea is to build a finer and finer partition of the nodes (by applying successive filters), and mapping nodes whenever possible to build the two isomorphisms considered -- if they exist: $\varphi$ and $f$.

An example of execution of the algorithm can be found in Appendix~\ref{annex:example}.

\medskip
\noindent
\textbf{Depth} We partition the only bag $B=(T_1,T_2)$, defining $T_i(d) = \lbrace u\in T_i : \depth(u)=d\rbrace$ for $d=0,\dots,\depth(T_i)$. We delete $B$ from $\bags$ and  for each $d$, we create a new bag $(T_1(d),T_2(d))$. Then, apply Deduction Rule~\ref{rule:mapping}. Note that since \textsc{SplitChildren} modifies bags after mapping two nodes, the number of bags meeting the prerequisite of the mapping deduction rule can vary through the iterations. At this step, since the roots are the only nodes with depth of 0, they must be mapped together, and the deduction rule is then applied at least once.

\medskip
\noindent
\textbf{Parents and children signature}
For each bag $B=(S_1,S_2)$ in $\bags$, we partition $S_1$ and $S_2$ by shared parent, i.e. we define $S_i(v) = \lbrace u \in S_i : \parent(u)=v\rbrace$. For any such a parent $v$, we define its children signature $\sigma(v)$ as the multiset $\sigma(v) = \lbrace [u] : u\in \children(v)\rbrace$. Nodes from $S_1(v)$ and $S_2(v')$ should be mapped together only if $\sigma(v)=\sigma(v')$. We then group the nodes by signature -- losing at the same time the parent information, but which will be recovered through the function \textsc{MapNodes} -- and define $S_i(s) = \cup_{\sigma(v)=s} S_i(v)$. We then create new bags $(S_1(s),S_2(s))$ for each such $s$, and finally delete $B$.

Once all bags have been partitioned, apply again Deduction Rule~\ref{rule:mapping}.

\medskip
\noindent
\textbf{Equivalence class under $\equiv$}
For each remaining bag $B=(S_1,S_2)$ in $\bags$, we partition $S_1$ and $S_2$ by equivalence class under $\equiv$, i.e. we define $S_i(c) = \lbrace u \in S_i : [u] = c\rbrace$. We then create new bags $(S_1(c),S_2(c))$ for each such $c$, and finally delete $B$.

Once all bags have been partitioned, apply again Deduction Rule~\ref{rule:mapping}.

\medskip
\noindent
\textbf{Labels} For each remaining bag $B=(S_1,S_2)$ in $\bags$, we now look at the labels of nodes in $S_1$ and $S_2$. We define $S_i(a) = \lbrace u\in S_i : \overline{u} = a\rbrace$. Some of these labels may have been seen previously and may be already mapped in $f$, in which case we can form bags with the related sets $S_i(a)$. Formally, we apply the following deduction rule.

\begin{deducrule}\label{rule:attr}
While there exist two sets (of same cardinality) $S_1(a)$ and $S_2(b)$ with $f(a)=b$, create bag $(S_1(a), S_2(b))$. If only one of the two sets exists ($S_1(a)$ with $a\in D_f$ or $S_2(b)$ with $b\in I_f$) but not its counterpart, we can conclude that $T_1\not\sim T_2$ and stop.
\end{deducrule}

The remaining $S_i(a)$ are to be mapped together. However, since we do not know the mapping between their labels, we cannot yet regroup them in bags. We create instead a collection $C$ that contains all those $S_i(a)$, and delete bag $B$.

Once all bags have been partitioned, either in new bags or in collections, apply Deduction Rule~\ref{rule:collection}. Since this rule maps new labels between them, new bags may be created by virtue of Deduction Rule~\ref{rule:attr}. Consequently, Deduction Rule~\ref{rule:attr} should be applied every time a bag is created by Deduction Rule~\ref{rule:collection} -- including during the \textsc{SplitChildren} procedure. Finally, apply again Deduction Rule~\ref{rule:mapping}.

\section{Analysis of the algorithm}\label{sec:analysis}
The analysis presented here is based on theoretical considerations and numerical simulations of labeled trees. For several given $n$ and $\mathcal{A}$, we generated 500 couples $(T_1,T_2)$ as follows. To create $T_1$, we generate a random recurvise tree \cite{zhang2015number} of size $n$, and assign a label, randomly chosen from the alphabet $\mathcal{A}$, to each node. We build $T_2$ as a copy of $T_1$, before randomly shuffling the children of each node. In this case, $T_1\sim T_2$. To get $T_1\not\sim T_2$, we choose a node $u$ of $T_2$ at random and replace its label by another one, drawn among $\mathcal{A}(T_1)\setminus \lbrace \overline{u}\rbrace$ -- this is the most difficult case to determine if $T_1\not\sim T_2$. The results are gathered in Figs.~\ref{fig:time} and~\ref{fig:logratio} and discussed later in the section. Remarkably, in terms of computation times and combinatorial complexity, they seem to mostly depend on $n$, and not $\#\mathcal{A}$. 

\subsection{The algorithm is linear}
In spite of an intricate back and forth structure between nodes, bags and collections (notably through deduction rules and the \textsc{SplitChildren} procedure), our algorithm is linear, in the following sense.

\begin{proposition}
The number of calls to the function \textsc{MapNodes} is bounded by the size of the trees.
\end{proposition}
\begin{proof}
Each call to \textsc{MapNodes} strictly reduces by one, in each tree, the number of nodes remaining to be mapped -- and thus present among the bags and collections. As a result, \textsc{MapNodes} cannot be called more times than the total number of nodes -- including the recursive calls of \textsc{MapNodes} on the parents.
\end{proof}

It is important to note, however, that this does not guarantee the overall linearity of the algorithm. Indeed, the complexity of a call to \textsc{MapNodes} depends on the number of deductions that will be made, notably though the \textsc{SplitChildren} procedure. 

Nevertheless, it seems that this variation regarding the deductions is compensated globally, since experimentally, as shown in Fig.~\ref{fig:comptime_ok}, in the case $T_1\sim T_2$, it appears quite clearly that the total computation time for the preprocessing phase is linear in the size of the trees. In the case $T_1\not\sim T_2$, the algorithm allows to conclude negatively in a sublinear time on average -- as shown in Fig.~\ref{fig:comptime_nok}.

\begin{figure}[H]
\centering

\begin{minipage}{0.5\textwidth}
\begin{subfigure}[b]{\textwidth}
\includegraphics[width=\textwidth]{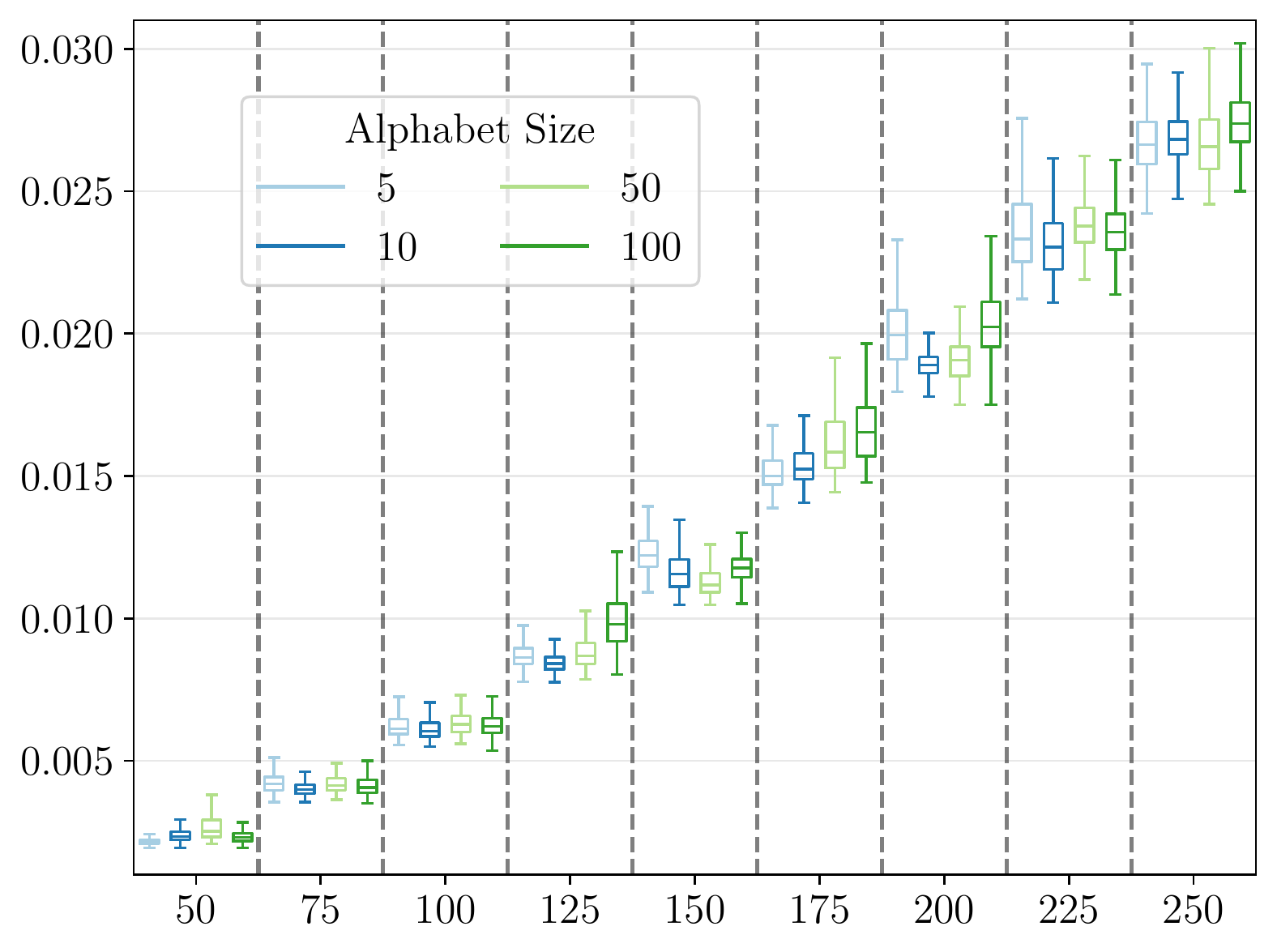}
\caption{\small Computation time when $T_1\sim T_2$.}
\label{fig:comptime_ok}
\end{subfigure}
\end{minipage}\hfill
\begin{minipage}{0.5\textwidth}
\begin{subfigure}[b]{\textwidth}
\includegraphics[width=\textwidth]{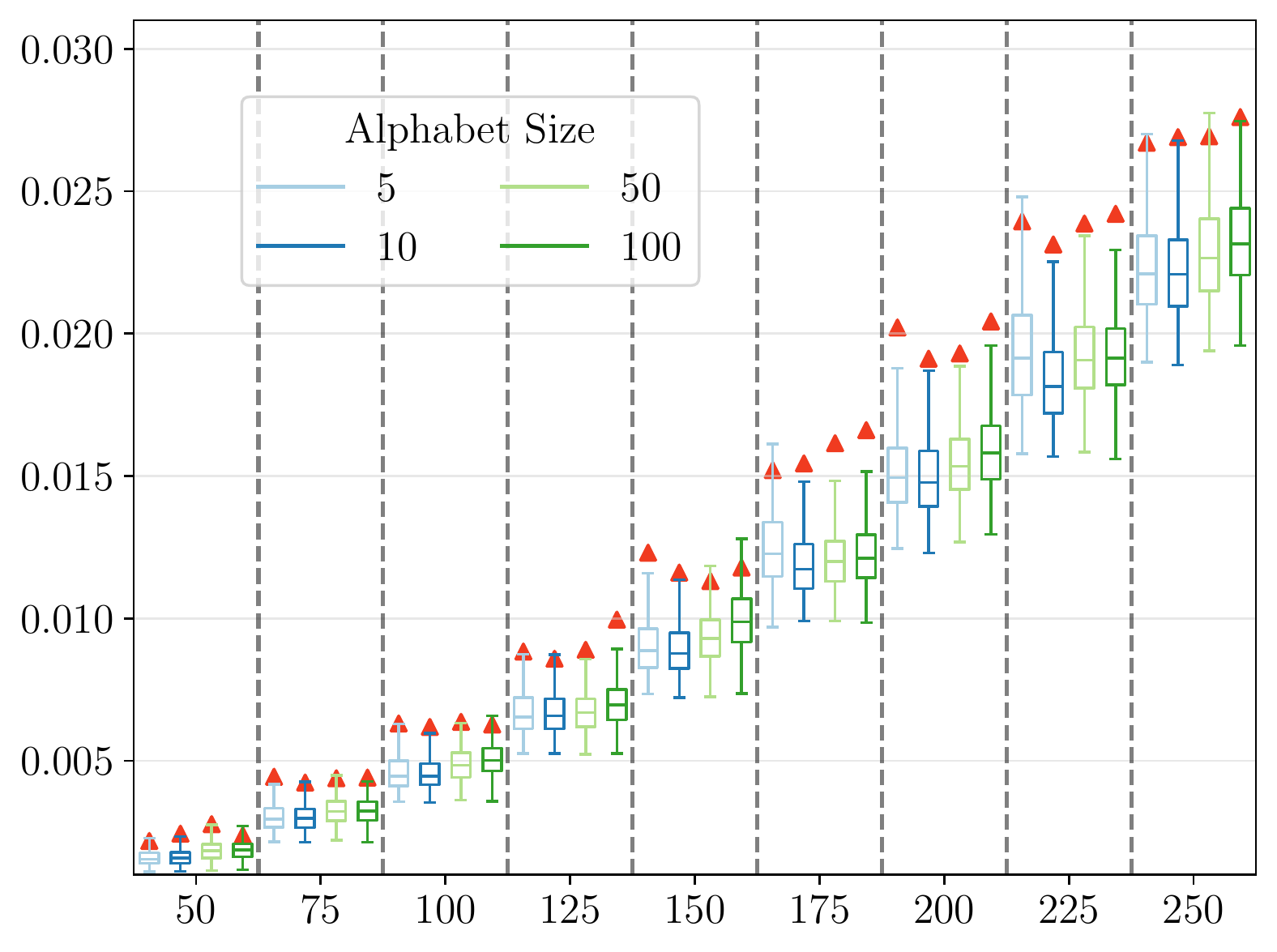}
\caption{\small Computation time when $T_1\not\sim T_2$.}
\label{fig:comptime_nok}
\end{subfigure}
\end{minipage}

\caption{\small Computation time (in s) for the execution of the algorithm of Subsection~\ref{ss:preprocessing}, according to the size of the considered trees. The different sizes of alphabet are displayed with different colors. In Fig.~\ref{fig:comptime_nok}, the red triangles indicate the average value of the corresponding computation time in the case $T_1\sim T_2$ (estimated from Fig.~\ref{fig:comptime_ok})}
\label{fig:time}
\end{figure}

\subsection{The algorithm reduces the complexity by an exponential factor on average}\label{ss:exp}
At any moment during the execution of the algorithm, given $\bags$ and $\collections$, we can deduce the current size of the search space. Indeed, for each bag $B$, there are $(\#B)!$ ways to map the nodes between them (not all of them necessarily leading to a tree isomorphism); for a collection $C$ and for given $n$, there are $(\#C(n))!$ ways to create bags, each giving $n!^{\#C(n)}$ possible mappings. The overall number of mappings associated to $C(n)$ is then given by $(n!)^{\#C(n)}(\#C(n))!$. Let us define the size of the current search space as
\[
N(\bags,\collections)=\prod_{B\in \bags} (\#B)! \prod_{C\in \collections} \left(\prod_n (n!)^{\#C(n)} (\#C(n))!\right)\]
Applying the deduction rules does not reduce this number at first sight -- since we transform into bags collections with $\#C(n)=1$ and we map nodes when $\#B=1$. On the other hand, each call to \textsc{SplitChildren} reduces this number. Indeed, for each bag or collection where a child of the mapped nodes appears, this object is divided into two parts, breaking the associated factorial:
\begin{itemize}
\item A bag with $(p+q)$ elements cut into two bags of size $p$ and $q$ reduces the size of the space by a factor of $\binom{p+q}{p}$.
\item An element of $C(p+q)$ cut into two elements of size $p$ and $q$ induces that $\#C(n)$ decreases by 1, and both $\#C(p)$ and $\#C(q)$ increase by 1. Overall, the size of the search space is modified by a factor of $\binom{p+q}{p}\frac{\#C(p+q)}{(\#C(p)+1)(\#C(q)+1)}$.
\end{itemize}
Each filter during the execution of the algorithm, that consists in splitting each bag into several ones has also the same effect on the overall cardinality. We can measure the evolution of the size of the search space by looking at the log-ratio $r(\bags,\collections)$, defined as follows -- with $N_\equiv(T_1)$ as in (\ref{eq:nb_isom}):

\[ r(\bags,\collections) = \log_{10} \frac{N(\bags,\collections)}{N_\equiv(T_1)}\]
The search space is reduced if and only if $r(\bags,\collections)$ is a negative number. It should be noted that we start the algorithm with a space size of $(\#T_1)!$, i.e. much more than $N_\equiv(T_1)$: the initial log-ratio is then positive. Note that despite having an initial search space bigger than $\isom(T_1,T_2)$, the algorithm cannot build a bijection that is not a tree isomorphism. The first topological filters (depth, parents, equivalence class) bring the log-ratio close to $0$ -- as illustrated in Fig.~\ref{fig:permutations} with 500 replicates of random trees of size 100 and an alphabet of size 5.

In more details, if we denote by $r_{\text{final}}(\bags,\collections)$ the log-ratio after the last filter on labels, Fig.~\ref{fig:logratio} provides a closer look at the results, and we can see that apart from pathological exceptions obtained with small trees, the log-ratio is always a negative number, so the algorithm does reduce the search space.

\noindent
\begin{minipage}[t]{0.5\textwidth}
\captionsetup{width=0.8\textwidth}
\includegraphics[width=\textwidth]{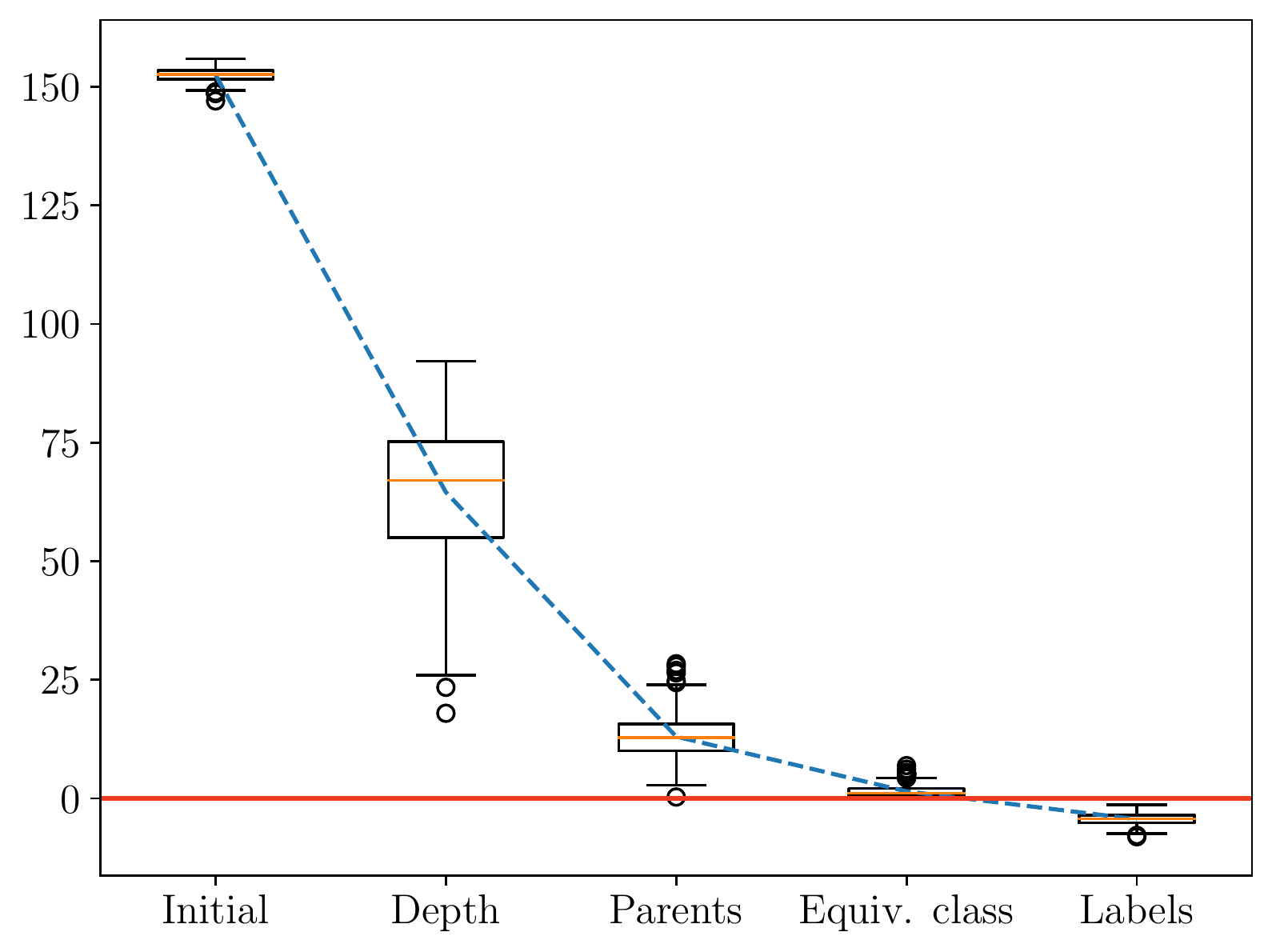}
\captionof{figure}{\small Evolution of $r(\bags,\collections)$ when $T_1\sim T_2$.}
\label{fig:permutations}
\end{minipage}\hfill
\begin{minipage}[t]{0.5\textwidth}
\includegraphics[width=\textwidth]{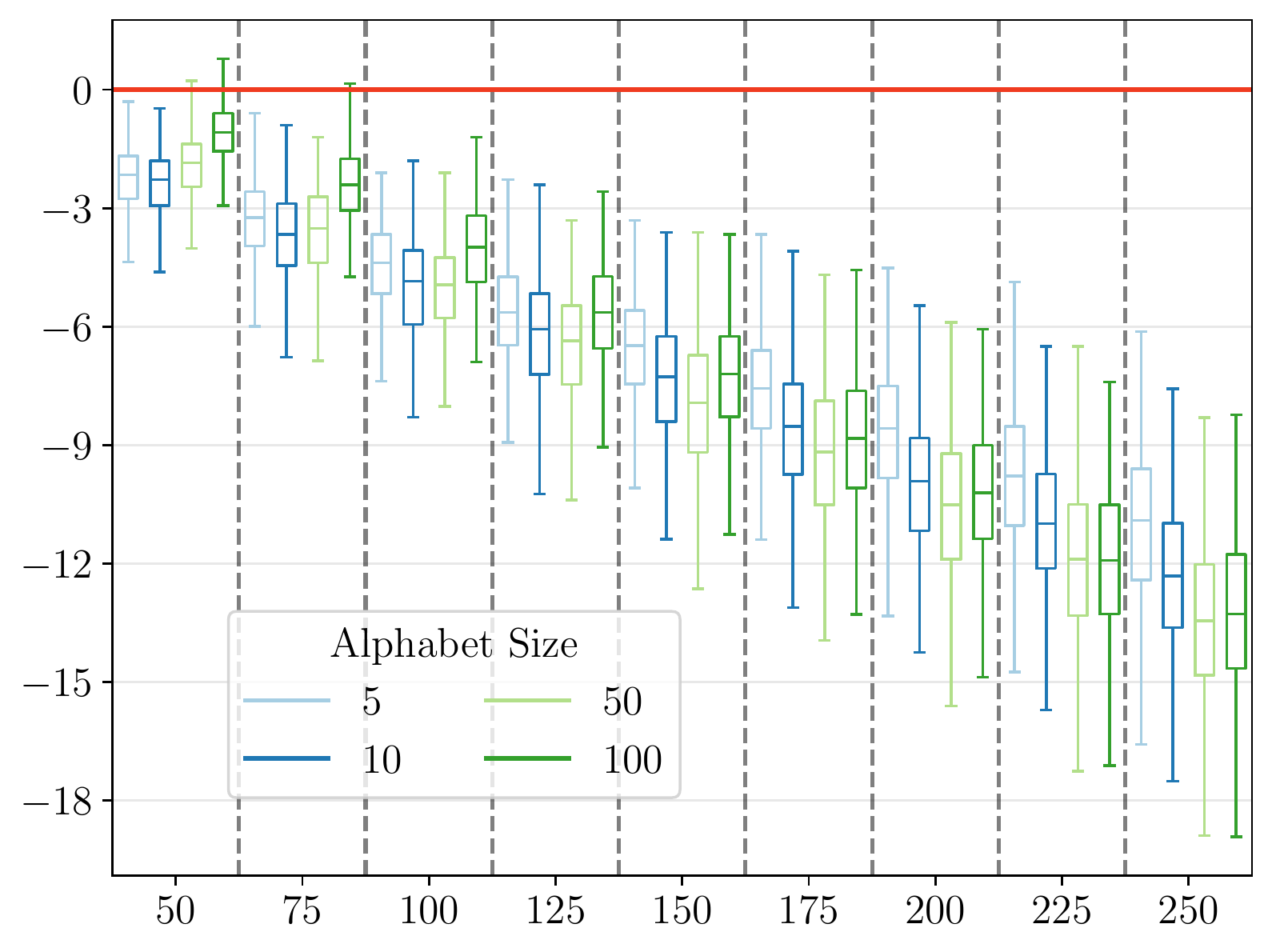}
\captionof{figure}{\small $r_{\text{final}}(\bags,\collections)$ when $T_1\sim T_2$,  according to the size of the considered trees. The different sizes of alphabet are displayed with different colors.}
\label{fig:logratio}
\end{minipage}

\medskip
\noindent
\textbf{As a conclusion,} we observe that the search space is reduced on average of an exponential factor and that this factor seems linear in the size of the tree. In other words, it seems that the larger the trees considered, the more exponentially the search space is reduced -- which is a remarkable property and justifies the interest of our method, especially given its low computational cost. 

\medskip
\noindent
\textbf{Implementation}  The algorithm presented in this paper has been implemented as a module of the Python library \verb=treex= \cite{azais2019treex}.

\medskip
\noindent
\textbf{Acknowledgements} The authors would like to thank three anonymous reviewers for their valuable comments on the first version of this manuscript.

\bibliographystyle{splncs04}
\bibliography{article}

\appendix

\section{Proof of Theorem~\ref{th:equiv_relation}}\label{annex:proof}

We begin with some preliminary reminders. Let $R$ be a relation over sets $E$ and $F$. $R$ is a bijection if and only if $\forall x\in E, \exists ! y\in F, x\mathrel{R}y$ and $\forall y\in E, \exists ! x\in E, x\mathrel{R}y$.

Let $R$ be a relation over sets $E$ and $F$; the converse relation $R^{-1}$ over sets $F$ and $E$ is defined as $y \mathrel{R^{-1}} x \iff x\mathrel{R}y$. If $R$ is a bijection, then so is $R^{-1}$.

Let $R$ be a relation over sets $E$ and $F$; and $S$ a relation over sets $F$ and $G$. The composition of $R$ and $S$, denoted by $S\circ R$, is a relation over $E$ and $G$, and defined as  $x\mathrel{(S\circ R)}z \iff \exists y\in F, (x\mathrel{R}y) \wedge (y\mathrel{S}z)$. If $R$ and $S$ are bijections, then so is $S\circ R$.

We now begin the proof. Let $T_1, T_2$ and $T_3$ be trees such that $T_1\tf{\varphi} T_2$ and $T_2\tf{\psi} T_3$. It should be clear that trivially, $T_1\tf{\id} T_1$. We aim to prove the following:

\[ T_1\tf{\psi\circ\varphi} T_3 \text{ and } T_2 \tf{\varphi^{-1}} T_1.\]

First of all, it is trivial that $\psi\circ \varphi\in \isom(T_1,T_3)$. The proof then follows directly from the reminders above and the two following lemmas:

\begin{lemma}
$R_{\psi \circ \varphi} = R_\psi \circ \mathrel{R_\varphi}$.
\end{lemma}
\begin{proof}
Let $x\in\attr(T_1)$ and $z\in \attr(T_3)$. It suffices to show \[x\mathrel{R_{\psi \circ \varphi}}z \iff \exists y \in \attr(T_2), x\mathrel{R_\varphi} y \wedge y\mathrel{R_\psi} z.\]
\begin{enumerate}
\item[$\implies$] There exists $u\in T_1$ so that $x=\overline{u}$ and $z=\overline{(\psi\circ \varphi)(u)}$. Let $v=\varphi(u)$ and $y=\overline{v}$; then $\overline{u}\mathrel{R_\varphi}\overline{v}$, so $x\mathrel{R_\varphi} y$; similarly $\overline{v}\mathrel{R_\psi} \overline{\psi(v)}$ leads to $y\mathrel{R_\psi }z$.
\item[$\impliedby$] There exists $u\in T_1$ so that $\overline{u}=x$ and $y=\overline{\varphi(u)}$. Let $v=\varphi(u)$. As $y\mathrel{R_\psi}\overline{\psi(v)}$, then $\overline{\psi(v)}=z$ and it follows $x\mathrel{R_{\psi\circ \varphi}}z$.
\end{enumerate}
\end{proof}

\begin{lemma}
$\mathrel{R_\varphi}^{-1}=R_{\varphi^{-1}}$.
\end{lemma}
\begin{proof}
Let $x\in\attr(T_1)$ and $y\in\attr(T_2)$. It suffices to show $x\mathrel{R_\varphi} y \iff y\mathrel{R_{\varphi^{-1}}}x$.
\begin{enumerate}
\item[$\implies$] There exists $u\in T_1$ so that $x=\overline{u}$ and $y=\overline{\varphi(u)}$. Let $v=\varphi(u)$. Since $u=\varphi^{-1}(v)$, $y\mathrel{R_{\varphi^{-1}}}x$.
\item[$\impliedby$] There exists $v\in T_2$ so that $\overline{v}=y$ and $x=\overline{\varphi^{-1}(v)}$. Let $u=\varphi^{-1}(v)$. Since $v=\varphi(u)$, $x\mathrel{R_\varphi} y$.
\end{enumerate}
\end{proof}

\section{Example of execution of the algorithm of Section~\ref{sec:break}}\label{annex:example}

We illustrate here the algorithm presented in Section~\ref{sec:break} on an example, namely the trees of Fig.~\ref{fig:example}. In addition to detailed explanations for each filter operation, 
a summary of the process can be found in Table~\ref{tab:example} at the end of this section.

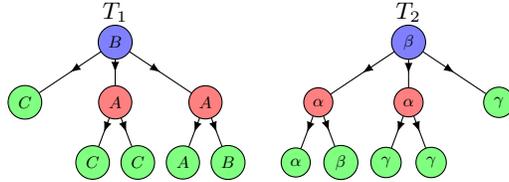
\begin{figure}[H]
\begin{minipage}{0.55\textwidth}
\def\xscale{0.3}
\def\yscale{0.4}
\def\nodescale{0.7}

\begin{tikzpicture}[xscale=\xscale,yscale=\yscale]
\tikzstyle{noeud}=[draw,circle,fill=white,scale=\nodescale*1]
\tikzstyle{attribut}=[scale=\nodescale*1,font=\bf]
\tikzstyle{arc}=[->-,>=latex]
\tikzstyle{fleche}=[->,>=latex,red,thick]

\node[noeud,fill=lblue] (u1) at (2,2) {$B$};
\node[noeud,fill=lred] (u2) at (6,0) {$A$};
\node[noeud,fill=lred] (u3) at (2,0) {$A$};
\node[noeud,fill=lgreen] (u4) at (5,-2) {$A$};
\node[noeud,fill=lgreen] (u5) at (7,-2) {$B$};
\node[noeud,fill=lgreen] (u6) at (1,-2) {$C$};
\node[noeud,fill=lgreen] (u7) at (3,-2) {$C$};
\node[noeud,fill=lgreen] (u8) at (-2,0) {$C$};

\node (t1) at (2,3) {$T_1$};

\draw[arc] (u1)--(u2) ;
\draw[arc] (u1)--(u3) ;
\draw[arc] (u1)--(u8) ;
\draw[arc] (u2)--(u4) ;
\draw[arc] (u2)--(u5) ;
\draw[arc] (u3)--(u6) ;
\draw[arc] (u3)--(u7) ;

\def\xshift{13}

\node[noeud,fill=lblue] (v1) at (2+\xshift,2) {$\beta$};
\node[noeud,fill=lred] (v2) at (-2+\xshift,0) {$\alpha$};
\node[noeud,fill=lred] (v3) at (2+\xshift,0) {$\alpha$};
\node[noeud,fill=lgreen] (v4) at (-3+\xshift,-2) {$\alpha$};
\node[noeud,fill=lgreen] (v5) at (-1+\xshift,-2) {$\beta$};
\node[noeud,fill=lgreen] (v6) at (1+\xshift,-2) {$\gamma$};
\node[noeud,fill=lgreen] (v7) at (3+\xshift,-2) {$\gamma$};
\node[noeud,fill=lgreen] (v8) at (6+\xshift,0) {$\gamma$};

\node (t1) at (2+\xshift,3) {$T_2$};

\draw[arc] (v1)--(v2) ;
\draw[arc] (v1)--(v3) ;
\draw[arc] (v1)--(v8) ;
\draw[arc] (v2)--(v4) ;
\draw[arc] (v2)--(v5) ;
\draw[arc] (v3)--(v6) ;
\draw[arc] (v3)--(v7) ;

\end{tikzpicture}
\end{minipage}~
\begin{minipage}{0.45\textwidth}
\caption{\small Two topologically isomorphic labeled trees $T_1$ and $T_2$. The color on nodes indicates the classes of equivalence of nodes under $\equiv$. Nodes are numbered from $u_1$ to $u_8$ in $T_1$ (resp. $v_1$ to $v_8$ in $T_2$) in breadth-first search order. As in Fig.~\ref{fig:cardequiv}, $N_\equiv(T_1)=8$.}
\label{fig:example}
\end{minipage}
\end{figure}

\medskip
\noindent
\textbf{Initialisation} We set $\varphi$ and $f$ as empty bijections and we create a single bag $B=\begin{pmatrix} u_1, \dots, u_8 \\ v_1,\dots, v_8\end{pmatrix}$. At this step, using the notation defined in Subsection~\ref{ss:exp}, $N(\bags,\collections)=8!=40,320$.

\medskip
\noindent
\textbf{Depth} We partition $B$ by considering the depth of the nodes. Since $\depth(T_1)=\depth(T_2)=2$, we create the following bags: 
\[B_0 = \begin{pmatrix} u_1\\ v_1\end{pmatrix}, B_1 = \begin{pmatrix} u_2, u_3, u_4\\ v_2, v_3, v_4\end{pmatrix} \text{ and } B_2 = \begin{pmatrix} u_5, u_6, u_7, u_8\\ v_5, v_6, v_7, v_8\end{pmatrix}.\]
Applying Deduction Rule~\ref{rule:mapping}, we call \textsc{MapNodes}$(u_1,v_1, \varphi,f)$ and delete $B_0$. Since the children of $u_1$ and $v_1$ already form a bag, the \textsc{SplitChildren} procedure does not divide any bags. After this step, we have $N(\bags,\collections) = 3! \times 4! = 144$, hence a reduction of the remaining space by a factor 280.

\medskip
\noindent
\textbf{Parents and children signature} Since the elements of bag $B_1$ all share the same parent, nothing happens here. However, let us look at bag $B_2$. We define the following sets $S_1(u_3)=\lbrace u_5, u_6\rbrace$, $S_1(u_4)=\lbrace u_7, u_8)$, $S_2(v_2) = \lbrace v_5, v_6\rbrace$ and $S_2(v_3)=\lbrace v_7, v_8\rbrace$. It appears that all those parents $u_3, u_4, v_2$ and $v_3$ have the same children signature $s = \lbrace \node[green], \node[green]\rbrace$. Therefore, the bag $B_2$ is rebuilt identically.

\medskip
\noindent
\textbf{Equivalence class under $\equiv$} The nodes of $B_2$ all share the same equivalence class \node[green] so the bag remains still. On the other hand, bag $B_1$ is splitted into 
\[ B_{\node[green]}=\begin{pmatrix} u_2\\v_4\end{pmatrix} \text{ and } B_{\node[red]} =\begin{pmatrix} u_3,u_4\\ v_2,v_3\end{pmatrix}.\]
Applying Deduction Rule~\ref{rule:mapping}, we call \textsc{MapNodes}$(u_2, v_4,\varphi, f)$ and delete bag $B_{\node[green]}$. Since the mapped nodes are leaves, there are no children to split, and their parents are already mapped. We then have $N(\bags,\collections) = 48$ and the remaining space has been reduced by 3.

\medskip
\noindent
\textbf{Labels} Here is what happens to each of the remaining bags:
\begin{description}
\item[$B_{\node[red]}$:] We create two sets $S_1(A) = \lbrace u_3,u_4 \rbrace$ and $S_2(\alpha) =\lbrace v_2,v_3\rbrace$. Since Deduction Rule~\ref{rule:attr} cannot be applied, we create a collection $C$ with $C(2) = \begin{pmatrix} \lbrace u_3, u_4 \rbrace \\ \lbrace v_2, v_3 \rbrace\end{pmatrix}$.
\item[$B_2$:] We create the following sets: $S_1(A) = \lbrace u_7\rbrace$, $S_1(B) = \lbrace u_8\rbrace$, $S_1(C) = \lbrace u_5, u_6\rbrace$; $S_2(\alpha) = \lbrace  v_5\rbrace$, $S_2(\beta) = \lbrace v_6\rbrace$ and $S_2(\gamma) = \lbrace v_7, v_8\rbrace$. Deduction Rule~\ref{rule:attr} allows to create bags $B_\beta = (S_1(B), S_2(\beta))$ and $B_\gamma = (S_1(C), S_2(\gamma))$. Finally, we create a collection $C'$ with $C'(1) = (S_1(A), S_1(\alpha))$.
\end{description}

After this step, we have the following bags and collections:

\[ B_\beta = \begin{pmatrix} u_8 \\ v_6\end{pmatrix}, B_\gamma = \begin{pmatrix} u_5, u_6 \\ v_7, v_8\end{pmatrix}, C : 2\mapsto \begin{pmatrix} \lbrace u_3, u_4 \rbrace \\ \lbrace v_2, v_3 \rbrace\end{pmatrix}\\ \text{and } C' : 1 \mapsto \begin{pmatrix} \lbrace u_7 \rbrace \\ \lbrace v_5 \rbrace\end{pmatrix}\]
where for collections, only the integers for which $\#C(n)>0$ are given. Applying Deduction Rule~\ref{rule:collection}, $C$ and $C'$ are deleted since $\#C(2)=1$ and $\#C'(1)=1$. We call \textsc{ExtBij}$(A,\alpha,f)$ and then the bags are:

\[ B_\beta = \begin{pmatrix} u_8 \\ v_6\end{pmatrix}, B_\gamma = \begin{pmatrix} u_5, u_6 \\ v_7, v_8\end{pmatrix}, B_C = \begin{pmatrix} u_3, u_4 \\ v_2, v_3 \end{pmatrix}\text{and } B_{C'} =\begin{pmatrix}  u_7  \\  v_5 \end{pmatrix}.\]

Applying Deduction Rule~\ref{rule:mapping}, we call \textsc{MapNodes}$(u_8,v_6, \varphi,f)$, therefore their parents must be mapped and we call \textsc{MapNodes}$(u_4,v_2,\varphi,f)$. $B_C$ is reduced to $(u_3,v_2)$. Applying Deduction Rule~\ref{rule:mapping} to $B_C$ and $B_{C'}$ maps $u_3$ with $v_2$ and $u_7$ with $v_5$. In the end, only $B_\gamma$ remains and therefore $N(\bags,\collections)=2!=2$, hence a reduction of a factor 24 of the remaining space.

The algorithm stops there; Fig.~\ref{fig:example_fin} illustrates the state of the bijections $\varphi$ and $f$ at the end of the execution.

\begin{table}
\centering
\begin{tabular}{c|c|c|c|c|c}
Filter & $\varphi$ & $f$ & $\bags$ & $\collections$ & $N(\bags,\collections)$ \\
\hline
Inititial & $\emptyset \mapsto \emptyset$ & $\emptyset \mapsto \emptyset$ & $\begin{pmatrix} u_1, u_2, u_3, u_4, u_5, u_6, u_7, u_8 \\
v_1, v_2, v_3, v_4, v_5, v_6, v_7, v_8
\end{pmatrix}$ & $\emptyset$ & $8!$ \\
\hline
Depth & $\begin{cases}u_1\mapsto v_1\end{cases}$ & $\begin{cases}B\mapsto \beta\end{cases}$  & $\begin{pmatrix} u_2, u_3, u_4\\
v_2, v_3, v_4
\end{pmatrix}$ ; $\begin{pmatrix}u_5, u_6, u_7, u_8 \\
v_5, v_6, v_7, v_8
\end{pmatrix}$ & $\emptyset$ & $3! \times 4!$\\
\hline
Parents &  \multicolumn{5}{c}{No changes}  \\
\hline
Equiv. class & $\begin{cases} u_1\mapsto v_1 \\ u_2\mapsto v_4
\end{cases}$ & $\begin{cases} B \mapsto \beta \\ C \mapsto \gamma
\end{cases} $ &$\begin{pmatrix} u_3, u_4\\
v_2, v_3
\end{pmatrix}$ ; $\begin{pmatrix}u_5, u_6, u_7, u_8 \\
v_5, v_6, v_7, v_8
\end{pmatrix}$ & $\emptyset$ & $2! \times 4!$ \\
\hline
Labels & $\begin{cases} u_1\mapsto v_1 \\ u_2\mapsto v_4\\ u_3\mapsto v_3\\ u_4\mapsto v_2\\ u_7\mapsto v_5\\ u_8 \mapsto v_6\end{cases}$ & $\begin{cases} B \mapsto \beta \\ C \mapsto \gamma \\ A \mapsto \alpha\end{cases} $ &  $\begin{pmatrix}u_5, u_6 \\
v_7, v_8
\end{pmatrix}$  & $\emptyset$ & $2!$ \\
\end{tabular}
\caption{\small Summary of the state of the different components of the problem at the end of each filter during the execution of the algorithm.}
\label{tab:example}
\end{table}

\begin{figure}[H]
\begin{minipage}{0.55\textwidth}
\centering
\def\xscale{0.3}
\def\yscale{0.5}
\def\nodescale{0.7}

\begin{tikzpicture}[xscale=\xscale,yscale=\yscale]
\tikzstyle{noeud}=[draw,circle,fill=white,scale=\nodescale*1]
\tikzstyle{attribut}=[scale=\nodescale*1,font=\bf]
\tikzstyle{arc}=[->-,>=latex]
\tikzstyle{fleche}=[->,>=latex,red,thick]

\node[noeud,draw=lblue, ultra thick] (u1) at (2,2) {$B$};
\node[noeud,draw=lgreen, ultra thick] (u2) at (6,0) {$A$};
\node[noeud,draw=lgreen, ultra thick] (u3) at (2,0) {$A$};
\node[noeud,draw=lgreen, ultra thick] (u4) at (5,-2) {$A$};
\node[noeud,draw=lblue, ultra thick] (u5) at (7,-2) {$B$};
\node[noeud] (u6) at (1,-2) {$C$};
\node[noeud] (u7) at (3,-2) {$C$};
\node[noeud,draw=lred, ultra thick] (u8) at (-2,0) {$C$};

\node (t1) at (2,4) {$T_1$};

\draw[arc] (u1)--(u2) ;
\draw[arc] (u1)--(u3) ;
\draw[arc] (u1)--(u8) ;
\draw[arc] (u2)--(u4) ;
\draw[arc] (u2)--(u5) ;
\draw[arc] (u3)--(u6) ;
\draw[arc] (u3)--(u7) ;

\node at (8.5,6) {$\varphi$};

\draw[ultra thick,dashed] (8.5,5.5)--(8.5,-3.5);

\def\xshift{13}

\node[noeud,draw=lblue, ultra thick] (v1) at (2+\xshift,2) {$\beta$};
\node[noeud,draw=lgreen, ultra thick] (v2) at (-2+\xshift,0) {$\alpha$};
\node[noeud,draw=lgreen, ultra thick] (v3) at (2+\xshift,0) {$\alpha$};
\node[noeud,draw=lgreen, ultra thick] (v4) at (-3+\xshift,-2) {$\alpha$};
\node[noeud,draw=lblue, ultra thick] (v5) at (-1+\xshift,-2) {$\beta$};
\node[noeud] (v6) at (1+\xshift,-2) {$\gamma$};
\node[noeud] (v7) at (3+\xshift,-2) {$\gamma$};
\node[noeud,draw=lred, ultra thick] (v8) at (6+\xshift,0) {$\gamma$};

\node (t1) at (2+\xshift,4) {$T_2$};

\draw[arc] (v1)--(v2) ;
\draw[arc] (v1)--(v3) ;
\draw[arc] (v1)--(v8) ;
\draw[arc] (v2)--(v4) ;
\draw[arc] (v2)--(v5) ;
\draw[arc] (v3)--(v6) ;
\draw[arc] (v3)--(v7) ;

\draw[-,draw=lblue, ultra thick] (u1)--(v1);
\draw[-,draw=lred, ultra thick] (u8) to [bend right=-45] (v8);
\draw[-,draw=lgreen, ultra thick] (u2)--(v2);
\draw[-,draw=lgreen, ultra thick] (u3) to [bend right=-15] (v3);
\draw[-,draw=lgreen, ultra thick] (u4) to [bend right=30] (v4);
\draw[-,draw=lblue, ultra thick] (u5) to [bend right=30] (v5);

\end{tikzpicture}
\end{minipage}~
\begin{minipage}{0.45\textwidth}
\caption{\small State of $\varphi$ (left) and $f$ (below) at the end of the algorithm.}
\label{fig:example_fin}
\centering
\def\xscale{0.3}
\def\yscale{0.5}
\def\nodescale{0.7}
\begin{tikzpicture}[xscale=\xscale,yscale=\yscale]
\tikzstyle{noeud}=[draw,circle,fill=white,scale=\nodescale*1]
\tikzstyle{attribut}=[scale=\nodescale*1,font=\bf]
\tikzstyle{arc}=[->-,>=latex]
\tikzstyle{fleche}=[->,>=latex,red,thick]

\def\xshift{0}

\node at (2+\xshift,4) {$f$};

\node[noeud,draw=lgreen, ultra thick] (A) at (\xshift,2) {$A$};
\node[noeud,draw=lblue, ultra thick] (B) at (2+\xshift,2) {$B$};
\node[noeud,draw=lred, ultra thick] (C) at (4+\xshift,2) {$C$};

\node[noeud,draw=lgreen, ultra thick] (a) at (\xshift,0) {$\alpha$};
\node[noeud,draw=lblue, ultra thick] (b) at (2+\xshift,0) {$\beta$};
\node[noeud,draw=lred, ultra thick] (c) at (4+\xshift,0) {$\gamma$};

\draw[-,draw=lblue, ultra thick] (B)--(b);
\draw[-,draw=lgreen, ultra thick] (A)--(a);
\draw[-,draw=lred, ultra thick] (C)--(c);

\end{tikzpicture}
\end{minipage}
\end{figure}
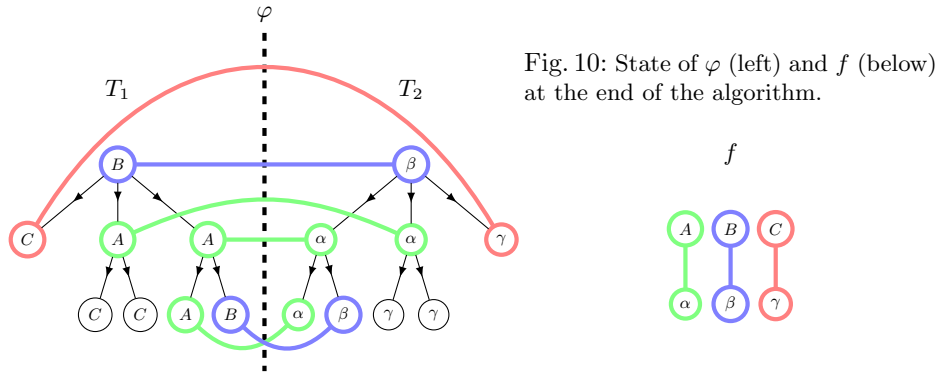

\end{document}